\documentclass[11pt]{article}
\usepackage{amsfonts, amssymb, amsmath, mathrsfs, amsthm, float, verbatim, graphicx}
\usepackage[hang,flushmargin]{footmisc}
\usepackage[margin=1.0in]{geometry}
\graphicspath{ {./images/} }

\newtheorem*{definition}{Definition}

\newtheorem{theorem}{Theorem}[subsection]
\newtheorem*{corollary}{Corollary}
\newtheorem*{hypothesis}{Hypothesis}
\title{Towards a Theory of Pragmatic Information}
\author{Edward D. Weinberger\\Department of Finance and Risk Engineering\\NYU Tandon School of Engineering\\edw2026@nyu.edu}
\begin{document}
\maketitle
\abstract{
Standard information theory says nothing about how much meaning is conveyed by a message. We fill this gap with a quantitative definition of ``pragmatic information'', the amount of meaning in a message relevant to a particular decision, a definition that is close to, but not identical with previous definitions. Here, we posit that such a message updates a random variable, $\omega$, that informs the decision. The pragmatic information of a single message is then defined as the Kulback-Leibler divergence between the prior and posterior probabilities of $\omega$; the pragmatic information of a message ensemble is the expected value of the pragmatic information of the ensemble's component messages. Canonical standard information theory results justify these definitions, most notably that that the pragmatic information of a single message is the expected difference between the shortest binary encoding of $\omega$ under the {\it a priori} and {\it a posteriori} distributions. We also show that the average of the pragmatic values of individual messages, when sampled a large number of times from the ensemble, approaches its expected value.\\

Pragmatic information is non-negative and additive for independent decisions and ``pragmatically independent'' messages. Also, pragmatic information is the information analogue of free energy: just as free energy quantifies the part of a system's total energy available to do useful work, so pragmatic information quantifies the information actually used in making a decision.\\

We sketch 3 applications: the single play of a slot machine, a.k.a. a ``one armed bandit'', with an unknown payout probability; a characterization of the rate of biological evolution in the so-called ``quasi-species'' model; and a reformulation of the efficient market hypothesis of finance. We note the importance of the computational capacity of the receiver in each case.
}

\pagebreak
\parindent 0px
\section{Introduction}
Just about any discussion of standard information theory notes that the standard entropy measure of the amount of information in a transmitted message says nothing about how meaningful the message is, even 
though a non-technical person would say that a message is ``informative'' if they are able to extract meaning from it.  Warren Weaver's introduction to Shannon’s foundational paper \cite{ShannonAndWeaver62} 
addresses this point by observing that the effectiveness of a communications process could be measured by answering any of the following three questions: \\
\begin{enumerate}
\item[A.] How accurately can the symbols that encode the message be transmitted (``the technical problem")?
\item[B.] How precisely do the transmitted symbols convey the desired meaning (``the semantics problem")?
\item[C.] How effective is the received message in changing conduct (``the effectiveness problem")?\\
\end{enumerate}
Weaver then observes that Shannon's paper --- and thus the entire edifice of what is now known as ``information theory" --- concerns itself only with the answer to above question A.\footnote{Arguably, however,{\it every} answer to question A is implicitly an answer to question C, as  every transmission implies that {\it something} in the receiver, which might be called its “conduct”, must have changed.}  The present paper, in contrast, is a continuation of previous work \cite{Weinberger02}, in which it was proposed that the only definitive way to assess the ``meaning” of question B above is to quantify the “changed conduct” of question C.\footnote{Such a quantification begs the question of whether this change in conduct is actually a change for the better!  Indeed, recent events have shown that dis-information can be quite effective in changing conduct, an observation that will be taken into account below.}\\

Although the primary goal of the present work is to advance the theory as such, any theory that is termed ``pragmatic'' should have practical applications, so we mention a few.  \cite{Weinberger02} proposed that the rate at which biological evolution aggregates ``meaningful information'' could be used to measure the rate at which evolution proceeds.  Throughout the present paper, we present a second application, namely
 how knowledge of previous payouts of a slot machine acquires meaning in the context of deciding whether to keep playing.  A final application is a reformulation of the celebrated efficient market hypothesis of financial economics as the claim that none of the available information about a given security is ``pragmatically useful''.\\

In the next section, we discuss the pragmatic information of a single message, preparing our quantitative results with a discussion of the properties that pragmatic information should have.
Some of these properties have been noted by other authors, qualitatively by such authors as \cite{Weizaeker} and \cite{Gernert}
and more quantitatively by \cite{Jumarie} and \cite {wow}, the latter having proposed definitions of pragmatic information similar to ours.  However, they do not provide anything like our demonstration that this definition follows from natural 
assumptions about what a ``measure of meaning'' should be.\\

We then consider the pragmatic information of messages in an ensemble.  As part of that discussion, we provide an example of how our definition applies to 
the simplest multi-armed bandit problem, that of playing a single slot machine with an unknown payout probability.  We also consider what it means for messages to be ``pragmatically independent'' and
the relationship between the pragmatic information of an ensemble of messages and the mutual information between the message ensemble and the ``state of the world'' variable $\omega$.  \\

We conclude with speculation about how a kind of ``computational capacity'' may play a role in the theory of pragmatic information similar to that of channel capacity in the Shannon theory, in that beliefs about $\omega$ cannot be updated during the non-zero time required to complete message processing.  Given that a participant in the financial markets can only find market information useful if it is meaningful in our sense, we also
discuss our ideas in the context of the efficient market hypothesis.

\section{The Pragmatic Information of a Single Message}
\subsection{Our Conceptual Framework}
As with the Shannon theory of information, we define a {\it message} to be a finite symbol string in some alphabet, albeit with the additional implication that meaning is being communicated.   
If the practical meaning of such a message stems from its usefulness in making a decision, the message must do so by changing those beliefs about the state of the world that will inform the subsequent action.  Accordingly, we consider a decision maker, $\Delta$, which may be a machine, a human, or a non-human organism.  $\Delta$'s decision is informed by its limited knowledge of the state of the world.  We represent this limited knowledge as the probability distribution of a random variable $\omega$ in some discrete sample space $\Omega =\{\omega_1, \omega_2, \ldots, \omega_i,\ldots, \omega_N \}$.  Initially, 
$\Delta$ assigns {\it a priori} probabilities  $\mathbf q = ( q_1, q_2, \ldots, q_N)$, each of which are positive, to the various possible values of $\omega$.  We
refer to $\mathbf q$ as {\it a priori} probabilities because we then supply $\Delta$ with a ``message", $m$.  As a result, $\Delta$'s beliefs about the state of the world change from $\mathbf q$ to 
$\mathbf p_m = ( p_{1|m}, p_{2|m}, \ldots, p_{N|m})$.\\  

While our conceptual framework can be applied to the wealth of so-called ``multi-armed bandit problems" 
as described in, for example, \cite{Robbins}, \cite{exploit} and \cite{MultiarmedBanditBook}, 
we consider the following, even 
simpler problem:  we have to decide whether to pay \$1 to play a slot machine, also known as a ``one 
armed bandit".  We know that will win a \$2 payout with some 
probability $\pi$, or nothing, with probability $1 - \pi$ that remains constant from one play to the next.
We also know that this payout probability is independent of the others.  We have access to the outcomes 
of as many previous plays as we wish.  The problem is that we don't know the value of $\pi$!\\

Suppose we denote the outcome of the $t^{th}$ such play by the random variable $\mathbf 1_t$, with
$$
\mathbf 1_t = \begin{cases} 1 & \quad \text{if we win on the $t^{th}$ play} \\
                             0 & \quad \text{otherwise}
                     \end{cases}
$$
Clearly, we learn something about $\pi$, and thus the correct play/no play decision for subsequent trials from the realized values of the sequence $\mathbf 1_1, \mathbf 1_2, \ldots, \mathbf 1_T$ , but exactly how much?
If we have won $w$ times in the $T$ trials, Laplace's Rule of Succession \cite{succession} estimates $\pi$ as
\begin{align}
\hat \pi(w, T) = \frac{w+1}{T+2}   \label{eq: succession}.
\end{align}
$\hat \pi(w, T)$ is less and less sensitive to whether or not any individual trial results in a payout as $T$ increases.  Thus, we learn less and less about what we need to know from each additional value in this sequence; indeed, after a finite number of such values, we know essentially all we need to know to make our decision.\\
  
In contrast, the amount of Shannon entropy, 
$$
\mathcal H(\mathbf 1_t) = \pi \log_2 \pi + (1-\pi) \log_2 (1-\pi),
$$
per play remains constant from one play to the next, and, since the trials are independent, the total Shannon entropy 
after $T$ trials is $T\mathcal H(\mathbf 1_t) = T \left[\pi \log_2 \pi + (1-\pi) \log_2 (1-\pi)\right]$.  However, 
embedded in the Shannon information for $T$ trials is information about whether particular plays are wins or losses, 
which is more than we need to know in order to estimate $\pi$ using \eqref{eq: succession}.\\

\subsection{Defining the Pragmatic Information of a Single Message}
Informed by the above, we propose the following
\begin{definition}
The {\bf pragmatic information}, $D_\Delta(\mathbf p_m || \mathbf q)$,
 of a single message $m$ acting on $\Delta$ is\footnote{As is customary in writing about information theory, we assume base
 2 logarithms here and subsequently.}
\begin{align}
D_\Delta(\mathbf p_m || \mathbf q) = \sum_i  p_{i|m} \log \left(\frac{p_{i|m}}{q_i}\right), \label{eq: BasicDef}
\end{align}
with the usual convention that $p_{i|m} \log \left(\frac{p_{i|m}}{q_i}\right) = 0$ if $p_{i|m} =0$.  In other words,
$D_\Delta(\mathbf p_m || \mathbf q)$ is the Kullback-Liebler divergence of $\mathbf p_m$ and $\mathbf q$.
\end{definition}

As we will see in more detail below, our definition of pragmatic information incorporates many of the features postulated by \cite{Weizaeker} and \cite{Gernert}, such as:
\begin{itemize}
\item $D_\Delta(\mathbf p_m || \mathbf q) = 0$ if $m$ is already known to $\Delta$.
\item $D_\Delta(\mathbf p_m || \mathbf q) = 0$ if  if $\Delta$ cannot process $m$.
\item $D_\Delta(\mathbf p_m || \mathbf q)$ will depend on $\Delta$, and, more specificially, on the context in which $\Delta$ receives $m$.
\item For two different messages, $m$ and $m'$, 
$$
D_\Delta(\mathbf p_m || \mathbf q) = D_\Delta(\mathbf p_{m'} || \mathbf q)
$$
if $m$ and $m'$ always lead $\Delta$ to choose the same decision probabilities, even if $m$ and $m'$ have different lengths and/or content.
\item $D_\Delta(\mathbf p_m || \mathbf q)$ increases if $\Delta$ is able to recognize the increasing novelty of $m$ as $\mathbf p_m$ differs more from $\mathbf q$.  However, if $m$ is completely novel, in the sense that $\Delta$ will be unable to process it, $D_\Delta(\mathbf p_m || \mathbf q) = 0$. 
\end{itemize}

Our conceptual framework also allows us to look more carefully at the following example, introduced by \cite{Gernert}:
\begin{quote}
\it{Consider a cash
register which stores the data of all sales over the day. After closing time, the shopkeeper
is not interested in the list of numbers, but in the total. By performing the addition the
cash register produces pragmatic information, since that total comes closer to the user
requirements than the list of raw data would do.}
\end{quote}
In our framework, $\Delta$ would be the shopkeeper who, presumably, needs to make decisions based on the true sales volume for the day, which we take as the relevant random variable $\omega$.  $\Delta$ would have an {\it a priori} estimate of $\omega$, based, perhaps, on a perusal of the list of individual sales.  This distribution, the {\it a priori} probability distribution, $\mathbf q$,  
when updated upon learning the cash register output of total sales, $S_T$, for the day, is likely to become more peaked around a single value.  Thus, if
$\mathbf p_T$ is this updated distribution, $D_\Delta(\mathbf p_T || \mathbf q) > 0$.\\  

But everything depends on the shopkeeper's subjective beliefs!  The shopkeeper might believe that the buggy software of the cash register is less reliable than the shopkeeper's own manual addition, in which case $S_T$ will be ignored  and $D_\Delta(\mathbf p_T || \mathbf q) = 0$.\\

The formula that \cite{wow} proposed for ``surprise" is essentially the same as \eqref{eq: BasicDef} for a single message.  Below, we will generalize this definition to any $m$ in some ensemble of possible messages.  However, \cite{wow} left ambiguous whether to define this quantity as (in our notation) $D_\Delta(\mathbf p_m || \mathbf q)$ or $D_\Delta(\mathbf q || \mathbf p_m)$.  They prefer the latter, as it easier for them to compute intuition-building analytic results in specific cases; however, we prefer the former, as we assume -- perhaps naively! -- that $\mathbf p_m$ really does provide an improvement over $\mathbf q$.  Our preferred definition also makes more sense in subsequent sections of this paper, where we list some known/easily derivable properties\footnote{See \cite{CoverAndThomas} for proofs of all of the theorems in the next subsection, except for our ``Wrong Code Theorem'', which is an embellishment of their results, and the corollary to Theorem 2.3.4.} of $D_\Delta(\mathbf p_m || \mathbf q)$ that make it an appealing choice for a measure of pragmatic information of a single message.\\

Similar ideas appear in the work of \cite{Lindley} on``the informativeness of an experiment'', the textbook 
\cite{deGroot}, and the ``reference posteriors'' of \cite{Bernardo}.  However, each of these authors base 
their analysis on the mutual information between $\mathcal M$ and $\Omega$.  Per our Theorem 
\ref{thm:PragInfoDecomp} below, this mutual information coincides with pragmatic information only in 
the special case where $q_i = \sum_{m \in \mathcal M} p_{i,m}$.  We are interested in the more general 
case where $\mathbf q$ is simply wrong and must be corrected.

\subsection{Properties of the Pragmatic Information of a Single Message}

The fundamental justification for our definition of the pragmatic information of a single message is found in the following two theorems and a corollary to the second:
\begin{theorem}[Wrong Code Theorem]\label{WrongCode}

Suppose $L_{\mathbf q}(\omega)$ is the length, in bits, of the shortest binary code required to communicate that $\Delta$ has decided upon outcome $\omega$, assuming the prior probabilities $\mathbf q$, and suppose $L_{\mathbf p_m}(\omega)$ is the corresponding length, assuming the prior probabilities $\mathbf p_m$.  Then 
$$
{\mathcal E}\left[L_{\mathbf q}(\omega) -  L_{\mathbf p_m}(\omega)\right]  = D_\Delta(\mathbf p_m || \mathbf q),
$$
where ${\mathcal E}\left[L_{\mathbf q}(\omega) -  L_{\mathbf p_m}(\omega)\right] $ is  the expected length of  
$L_{\mathbf q}(\omega) -  L_{\mathbf p_m}(\omega)$ under the {\rm a posteriori} probabilities $\mathbf p_m$. 
\end{theorem}
\begin{proof}
Let ${\mathcal  H}(\mathbf p_m)$ be the Shannon entropy of  $\mathbf p_m$.  Then, per \cite{CoverAndThomas},
$$
{\mathcal H}(\mathbf p_m) + D(\mathbf p_m || \mathbf q) \leq {\mathcal E}\left[L_{\mathbf q}(\omega)\right]< {\mathcal H}(\mathbf p_m) + D_\Delta(\mathbf p_m || \mathbf q) + 1,
$$
and
$$
{\mathcal H}(\mathbf p_m) \leq {\mathcal E}\left[L_{\mathbf p_m}(\omega)\right] <{\mathcal H}(\mathbf p_m) + 1,
$$ 
The desired result follows upon subtracting the second set of inequalities from the first.
\end{proof}
We conclude that $D_\Delta(\mathbf p_m || \mathbf q)$ can reasonably be interpreted as the amount of information that $\Delta$ has ``learned'' from $m$, assuming that 
$\mathbf p_m$ really is a better estimate of the true probablilty distribution of $\omega$ than $\mathbf q$.\\

Next, consider two decision makers, $\Delta$ and $\Delta'$, with respective input messages $m$ and $m'$, estimating the likelihood of random variables $\omega \in \Omega$ and $\omega' \in \Omega'$, respectively.  Suppose:
\begin{itemize}
\item $\{q_{i,i'}\}$ is the set of{\it a priori} joint probabilities regarding outcome $\omega_i \in \Omega$ and, simultaneously, outcome $\omega'_{i'} \in \Omega'$, each without knowledge of respective input messages $m$ and $m'$.
\item the action of both messages changes $\{q_{i,i'}\}$ to the {\it a posteriori} joint probabilities, $\{p_{i,i'|m, m'}\}$.
\item $\{q'_{i'|i}\}$ is the set of prior probabilities of $\Delta'$'s decisions, given the decisions of $\Delta$, but not message $m$, and $\{p'_{i'|i,m,m'}\}$ is the {\it a posteriori} probability distribution of $\Delta'$'s decisions, given the decisions of $\Delta$ {\it and} both messages $m$ and $m'$.  

\end{itemize}
We then have 
\begin{theorem}[Chain Rule for Kullback-Leibler Divergence]
\label{DivergenceChainRule}
$$
D(\{p_{i,i'|m, m'}\} || \{q_{i,i'}\}) 
= D(\mathbf p_m || \mathbf q) + D(\{p'_{i'|i,m,m'}\} || \{q'_{i'|i}\})
$$
\end{theorem}

\begin{corollary}[Additivity of Kullback-Leibler Divergence]
\label{KLAdditivity}
In the special case where $\Delta$ and $\Delta'$ are completely independent of each other, {\rm i.e.} that $\{q'_{i'|i}\} = \mathbf q'$ and $\{p'_{i'|i,m,m'}\} =\mathbf p'_{m'}$, 
$$
D(\{p_{i,i'|m,m'}\} || \{q_{i,i'}\})  = D(\mathbf p_m|| \mathbf q) + D(\mathbf p'_{m'}|| \mathbf q'),
$$
\end{corollary}
about which Kullback writes in \cite{Kullback} that 
\begin{quote}
\it{``Additivity of information for independent events is intuitively a fundamental requirement, and is indeed postulated in most axiomatic developments of information theory.  Additivity is the basis for the logarithmic form of information.  A sample of $n$ independent observations from the same population {\rm [here $n$ copies of $\Delta$ making $n$ identical and independent decisions]} provides $n$ times the mean information in a single observation.''}
\end{quote}

$D_\Delta(\mathbf p_m || \mathbf q)$ has a few other properties that would seem to be required of a measure of meaning, as stated in the following theorems:

\begin{theorem}[Non-negativity of Kullback-Leibler Divergence]\label{NonNegativity}
$$
D_\Delta(\mathbf p_m || \mathbf q) \geq 0,
$$
with equality if and only if $\mathbf p_m = \mathbf q$. 
\end{theorem}
Thus, $m$ always conveys positive pragmatic information, unless $m$ is ignored by $\Delta$.\\

\begin{theorem}[Convexity of Kullback-Leibler Divergence]
\label{DivergenceConvexity}
Consider the pairs of probability distributions $\mathbf p_m$ and $\mathbf p'_m$ and $\mathbf q$, and $\mathbf q'$.  Form, for any $\lambda \in [0, 1]$, the pair of interpolated distributions 
$$
\mathbf p_m^\lambda = \lambda \mathbf p_m + (1 - \lambda)\mathbf p'_m
$$
and 
$$
\mathbf q^\lambda = \lambda \mathbf q + (1 - \lambda)\mathbf q',
$$
then
$$
D\Big(\mathbf p_m^\lambda || \mathbf q^\lambda\Big) \leq 
\lambda D\left(\mathbf p_m || \mathbf q \right)+ (1-\lambda) D\left(\mathbf p'_m || \mathbf q' \right).
$$
\end{theorem}
It is useful at this point to introduce the following
\begin{definition}
A message $m$ is {\bf pragmatically definitive} for a given $\Delta$ if $\mathbf p_m = \mathbf u_k$ for some $k$, {\it i.e.} the unit vector in the $k^{\rm th}$ direction.
\end{definition}
In other words, $\Delta$ is certain that outcome $k$ will be achieved if it receives a pragmatically definitive message to that effect.  
Note that all that pragmatic definitiveness guarantees is the definitive prediction of {\it an} outcome, but not necessarily the desired one!\\

Using this terminology, we have  for given  $\mathbf q$, the following
\begin{corollary}
The maximum value of $D_\Delta(\mathbf p_m || \mathbf q)$ is the maximum of all possible pragmatically definitive messages, 
namely $\max_i (-\log q_i)$.
\end{corollary}
\begin{proof}
Since $\mathbf p_m = \sum_i p_{i|m} \mathbf u_i$, where $\mathbf u_i$ is a unit vector in the ${\rm i}^{th}$ direction, it follows from repeated applications of Theorem 2.3.4 that 
\begin{align*}
D\Big(\mathbf p_m || \mathbf q\Big)          &=  D\Big(\sum_i p_{i|m} \mathbf u_i || \mathbf q\Big)                        \\
                                                              &\leq \sum_i p_{i|m} D\left(\mathbf u_i || \mathbf q \right)                   \\
                                                              &\leq \max_i \left(-\log q_i \right)
\end{align*}

\end{proof}
\section{The Pragmatic Information of an Ensemble of Messages}
We now assume that $\Delta$ is asked to make a series of independent decisions, each time given the same prior set of beliefs ({\it i.e.} $\mathbf q$ is the same each time), but each time given a (possibly different) 
message $M_k$ from some finite ensemble $\mathcal M$.  
We further assume that each $M_k$ is sampled using the same probability distribution, namely ${\rm Pr}\left\{M_k = m \right\} = \varphi_m$ for all $k$ and 
every $m \in \mathcal M$ is eventually sampled, regardless of how the sampling starts.  Given the additivity of Kullback-Leibler divergence for independent decisions, 
$M_k$ conveys the (possibly different) amount of pragmatic information $D_\Delta(\mathbf p_{M_k} || \mathbf q)$ for each $k$.  
The average amount of pragmatic information that $\Delta$ accrues per message is then the random variable
$$
\Phi_N = \frac{1}{N} \sum_{k=1}^N D_\Delta(\mathbf p_{M_k} || \mathbf q).
$$ 
Our assumptions on how the $M_k$'s are sampled are sufficient to guarantee that this sum satisfies Birkoff's ergodic theorem \cite{Billingsley}.  We can therefore conclude that
$$
\lim_{N \rightarrow \infty} \Phi_N = E \left[ D(\mathbf p_M || \mathbf q) \right],
$$
with probability 1, with the expectation taken under the distribution of the $M_k$'s.  
We stress that, as in the previous section, $\Delta$ makes each of its decisions as though none of the previous decisions had ever been made.  However, the hypotheses of the ergodic theorem are sufficiently weak 
that no such assumption is required for the $M_k$'s.  While they could be sampled independently, they could also be a realization of a Markov chain, etc.  The generality of this result therefore justifies the following 
\begin{definition}
The {\bf pragmatic information}, ${\Phi_\Delta (\mathcal M; \Omega)}$,  of an ensemble of messages ${\mathcal M}$ is the expected value 
\begin{align*}
 \Phi_\Delta (\mathcal M; \Omega) &= E \left[ D(\mathbf p_m || \mathbf q) \right] \\
                                      &= \sum_{i, m} \varphi_m  p_{i|m} \log \left(\frac{p_{i|m}}{q_i}\right) \\
                                     &= \sum_{i, m}  p_{i,m} \log \left(\frac{p_{i|m}}{q_i}\right)
\end{align*}
where $p_{i|m}$ remains the conditional probablilty of outcome $\omega_i$, given the receipt of message $m$, so that the joint probability, $p_{i,m}$, of outcome $\omega_i$ upon receiving message $m$ is $p_{i,m} = \varphi_m  p_{i|m}$.  Once again, we take $p_{i|m} \log \left(\frac{p_{i|m}}{q_i}\right) = 0$ if $p_{i|m} =0$.
\end{definition}

\subsection{An Example:  The One Armed Bandit Problem}
It is instructive to apply this definition to the ``one armed bandit problem" discussed above.  Evidently, the relevant decision is whether it is worth playing the slot machine, a decision that is informed by a random variable, $\omega$, that can assume 
exactly one of the two values $\mathtt{PAYOUT}$ and $\mathtt{NO PAYOUT}$, {\it i.e.} $\Omega = \{\mathtt{PAYOUT}, \mathtt{NO PAYOUT}\}$.  We want to know the pragmatic information accruing from a knowledge of $\mathbf 1_{T+1}$, given that the $T$ prior trials have resulted in $w$ wins.  Thus, per \eqref{eq: succession}, we take $\hat \pi(w, T)$ as $q_1$, our {\it a priori} estimate of the payout probability, and $1 - \hat \pi(w, T)$ as $q_0$, our {\it a priori} estimate of the no payout probability.\\  

The ensemble, $\mathcal M_{T+1}$, of possible messages for the $T+1^{st}$ trial is $\mathtt{PAYOUT}_{T+1}$ and $\mathtt{NO PAYOUT}_{T+1}$.  If $\Delta$ receives the $\mathtt{PAYOUT}_{T+1}$ message, 
the {\it a posteriori} $\mathtt{PAYOUT}$ and $\mathtt{NO PAYOUT}$ probabilities are updated to
\begin{align*}
\mathbf \pi_1 &= \Big(\hat \pi(w+1, T+1), 1 - \hat \pi(w+1, T+1)\Big) \cr
                  &= \Big(\frac{w+2}{T+3},\ \frac{T-w+1}{T+3}\Big) \cr
\end{align*}
Similarly, if $\Delta$ receives the $\mathtt{NOPAYOUT}_{T+1}$ message,
the corresponding {\it a posteriori} probability vector is
\begin{align*}
\mathbf \pi_0 &= \Big(\hat \pi(w+1, T+1), 1 - \hat \pi(w+1, T+1)\Big) \cr
                  &= \Big(\frac{w+1}{T+3},\ \frac{T-w+2}{T+3}\Big) \cr
\end{align*}
We then have 
\begin{align*}
D_\Delta(\mathbf \pi_1 || \mathbf q) 
        &= \hat \pi(w+1, T+1) \log\left[\frac{\hat \pi(w+1, T+1)}{\hat \pi(w, T)}\right] + \\
        &  \qquad    \left[1 - \hat \pi(w+1, T+1)\right] \log\left[\frac{1 - \hat \pi(w+1, T+1)}{1 -\hat \pi(w, T)}\right] \\[10pt]
        &=\frac{w+2}{T+3}\log\left[\frac{(w+2)(T+2)}{(T+3)(w+1)}\right] + \frac{T-w+1}{T+3}\log\left[\frac{T+2}{T+3}\right]
\end{align*}
and
\begin{align*}
D_\Delta(\mathbf \pi_0 || \mathbf q) 
     &= \hat \pi(w, T+1) \log\left[\frac{\hat \pi(w, T+1)}{\hat \pi(w, T)}\right] + \\
     &  \qquad    \left[1 - \hat \pi(w, T+1)\right] \log\left[\frac{1 - \hat \pi(w, T+1)}{1 -\hat \pi(w, T)}\right] \\[10pt]
     &=\frac{w+1}{T+3}\log\left[\frac{T+2}{T+3}\right]+\frac{T-w+2}{T+3}\log\left[\frac{(T-w+2)(T+2)}{(T+3)(T-w+1)}\right],
\end{align*}
both of which can be interpreted as the reduction in the uncertainty surrounding the estimate of $\pi$, upon receipt of the corresponding value of $\mathbf 1_{T+1}$.  We conclude that the pragmatic information provided by the $T+1^{st}$ trial is
\begin{align*}
\Phi_\Delta \left(\mathcal M_{T+1}; \Omega\right) 
    &= \pi D_\Delta(\mathbf \pi_1 || \mathbf q) + (1-\pi) D_\Delta(\mathbf \pi_0 || \mathbf q)\\[10pt]
    &=\pi \left\{\frac{w+2}{T+3}\log\left[\frac{(w+2)(T+2)}{(T+3)(w+1)}\right] + \frac{T-w+1}{T+3}\log\left[\frac{T+2}{T+3}\right]\right\}  \\
    &  \qquad + (1-\pi)\left\{\frac{w+1}{T+3}\log\left[\frac{T+2}{T+3}\right] + 
                                     \frac{T-w+2}{T+3}\log\left[\frac{(T-w+2)(T+2)}{(T+3)(T-w+1)}\right]\right\}.
\end{align*}

As $T$ gets large, $w = \pi T + o(T)$ with probability 1, by the Kolmogorov Strong Law of Large Numbers \cite{Feller}. 
 Straightforward algebraic manipulation then shows that the arguments of all of the logarithms
in $\Phi_\Delta \left(\mathcal M_{T+1}; \Omega\right)$ approach 1 monotonically as $T$ increases.  We conclude that $\Phi_\Delta \left(\mathcal M_{T+1}; \Omega\right)$ steadily decreases to zero with increasing $T$.
Figure 1 shows this decrease, with $w = \pi T$, its most likely value.  It also shows that the pragmatic information accruing from each trial is larger as $\pi$ gets closer to 1/2, as we would expect:  
the play/no play decision is less obvious in that case, so knowing the outcome of each additional previous trial is all the more informative.
\begin{figure}[ht]
    \centering
    \includegraphics[width=\textwidth]{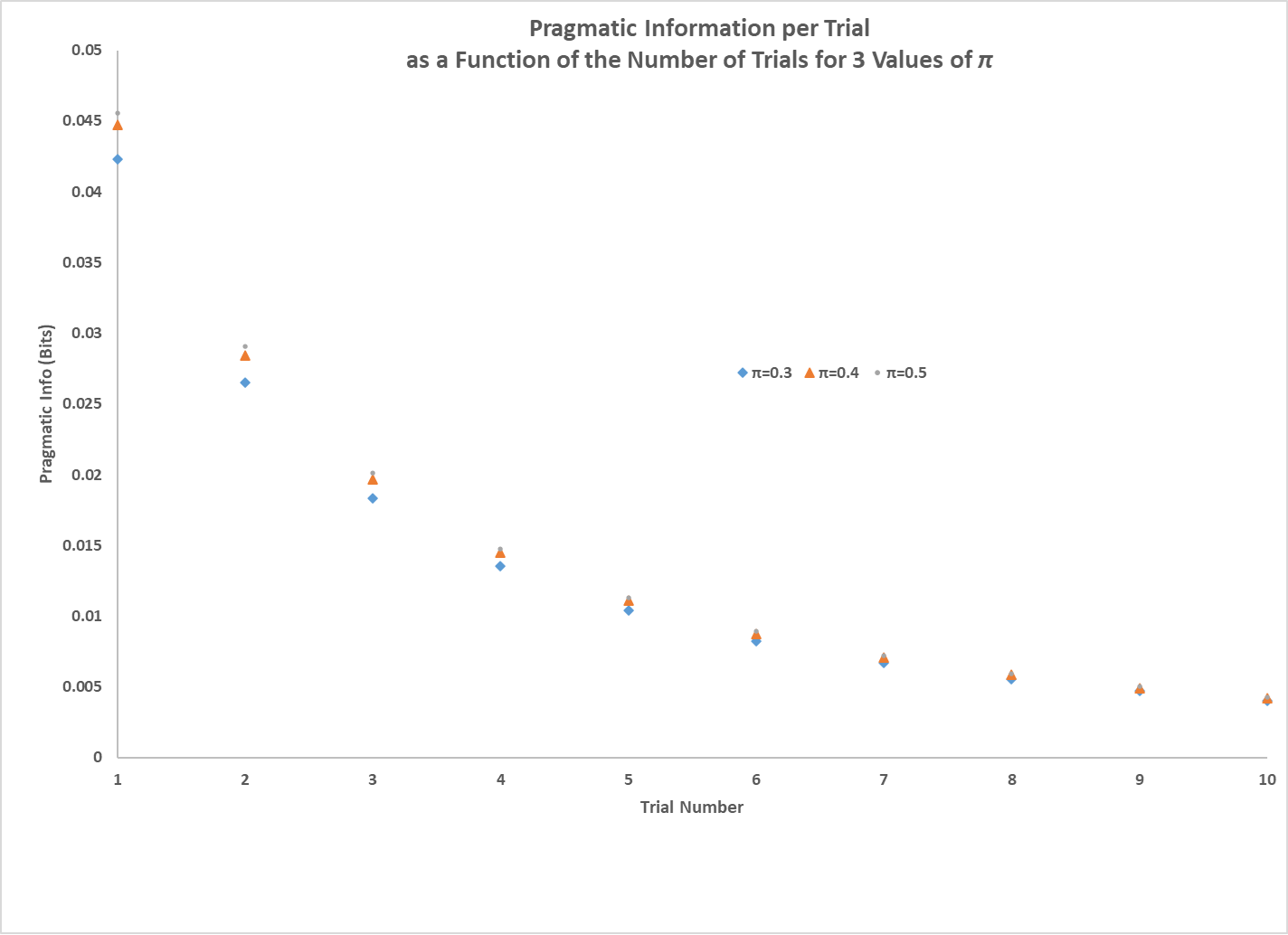} 
    \caption{Pragmatic information for the ``one armed bandit''}
    \label{fig:sample}
\end{figure}
\subsection{Properties of the Pragmatic Information of an Ensemble of Messages}
Because ${\Phi_\Delta (\mathcal M; \Omega)}$ is just a weighted sum of Kullback-Leibler divergences, we have immediately
\begin{theorem}[Non-negativity of Pragmatic Information]
$\Phi_\Delta (\mathcal M; \Omega) \ge 0$,
with equality if and only if none of the messages in $\mathcal M$ change $\Delta$'s {\it a priori} estimate of $\omega$.
\end{theorem}
The following definition will be useful in the discussion below:
\begin{definition}
A message ensemble $\mathcal M$ is {\bf pragmatically definitive} for a given $\Delta$ if for all $m \in \mathcal M$, $m$ is pragmatically definitive.
\end{definition}
For such ensembles, we have $p_{i|m} = \delta_{i, k(m)}$, {\it i.e.} $p_{i|m} =1$ if $i = k(m)$, the outcome to which $m$ is mapped deterministically and zero otherwise.  It follows that
\begin{align*}
 \Phi_\Delta (\mathcal M; \Omega) &= \sum_{i, m} \varphi_m  p_{i|m} \log \left(\frac{p_{i|m}}{q_i}\right) \\
                                                  &= \sum_{i, m} \varphi_m \delta_{i, k(m)} \log \left(\frac{\delta_{i, k(m)}}{q_i}\right) \\
                                                  &= - \sum_m \varphi_m \log q_{k(m)}
\end{align*}
for such ensembles.  Since the $\varphi$'s sum to 1, we also have
\begin{theorem}[Upper Bound on Pragmatic Information]
\begin{align*}
\Phi_\Delta (\mathcal M; \Omega) &\leq -\sum_m \varphi_m \log q_{k(m)} \\
                                                 &\leq \max_i \left(-\log q_i \right)\\
\end{align*}
\end{theorem}
and
\begin{theorem}[Wrong Code Theorem for Pragmatic Information]
$\Phi_\Delta (\mathcal M; \Omega)$ is the expected number of extra bits required to code samples from $\mathbf p_m$, averaged over $m \in \mathcal M$, using a binary encoding optimized for $\bf q$, rather than,
for each $m$, a binary encoding optimized for that $\mathbf p_m$.
\end{theorem}

Next, we generalize joint and conditional Kullback-Leibler Divergence.  Suppose:

\begin{itemize}
\item decision makers $\Delta$ and $\Delta'$ have respective ``state of the world'' variables $\omega \in \Omega$ and $\omega' \in \Omega'$,  
\item $\{q_{i,i'}\}$ is the set of {\it a priori} joint probabilities that $\Delta$ assigns to $\omega_i \in \Omega$ and $\Delta'$ assigns to $\omega'_{i'} \in \Omega'$, prior to receipt of respective input messages $m$ and $m'$,
\item the action of both messages changes the joint probabilities of outcomes $\omega_i$ and $\omega'_{i'}$ from the prior probabilities, $\{q_{i,i'}\}$, to the {\it a posteriori} joint probabilities, $\{p_{i,i'|m, m'}\}$, and 
\item $\{q'_{i'|i}\}$ is the set of prior probabilities of $\Delta'$'s decisions, given the decisions of $\Delta$, but not message $m$, and $\{p'_{i'|i,m,m'}\}$ is the {\it a posteriori} probability distribution of $\Delta'$'s decisions, given the decisions of $\Delta$ {\it and} both messages $m$ and $m'$.  
\item The joint probability that $\Delta$ receives message $m$ and $\Delta'$ receives message $m'$ is $\varphi_{m, m'}$.
\item The joint probability that these messages are received and the decision maker's output $\omega_i$ and $\omega'_{i'}$, respectively, is $p_{i, i', m, m'}$.  
\end{itemize}
We then have the following
\begin{definition}
The {\bf joint pragmatic information} of messages in ${\mathcal M}$ acting on $\Delta$ and messages in ${\mathcal M'}$ acting on $\Delta'$ is
\begin{align*}
\Phi_{\Delta, \Delta'} (\mathcal M, \mathcal M'; \Omega, \Omega') &= E \left[ D(\{p_{i,i'|m, ms'}\} || \{q_{i,i'}\}) \right] \\
                                     &= \sum_{i, i', m, m'} \varphi_{m, m'}  p_{i, i'|m, m'} \log \left(\frac{p_{i,i'|m, m'}}{q_{i,i'}}\right)\\                                 
                                     &= \sum_{i, i', m, m'} p_{i, i', m, m'} \log \left(\frac{p_{i,i'|m, m'}}{q_{i,i'}}\right)  \\
\end{align*}
\end{definition}
and  
\begin{definition}
The {\bf conditional pragmatic information} of messages in ${\mathcal M'}$ acting upon $\Delta'$, given that $m \in \mathcal M$ acting on $\Delta$, as well as the resulting outcome $\omega$ is
\begin{align*}
 \Phi_{\Delta'| \Delta} (\mathcal M'; \Omega' | \mathcal M, \Omega) &= E \left[D(\{p_{i'|i,M,M'}\} || \{q_{i'|i}\}) \right] \\
                                      &= \sum_{i, i', m, m'} p_{i, i', m, m'} \log \left(\frac{p'_{i'|i,m,m'}}{q_{i'|i}}\right).
\end{align*}
\end{definition}
Since the Kullback-Leibler divergence of each $m \in \mathcal M$ satisfies the Chain Rule for Kullback-Leibler divergence, the generalization of that Rule follows immediately:
\begin{theorem}[Chain Rule for Pragmatic Information]
\label{PragmaticInfoChainRule}
Suppose each $m \in \mathcal M$ and $m' \in \mathcal M'$ acts, respectively, on decision makers $\Delta$ and $\Delta'$,   with respective prior decision probabilities $\mathbf q$ and $\mathbf q'$ and with corresponding {\rm a posteriori} decision probabilities $\mathbf p_m$ and $\mathbf p'_{m'}$, then
$$
 \Phi_{\Delta, \Delta'} (\mathcal M; \Omega, \Omega')=\Phi_{\Delta} ({\mathcal M}; \Omega)  +
\Phi_{\Delta'| \Delta} (\mathcal M'; \Omega' | \mathcal M, \Omega)
$$
\end{theorem}

As an example, suppose $\Delta$ and $\Delta'$ have the
respective Boolean state of the world variables $\omega$ and $\omega'$, and, {\it a priori}, we know nothing about them.  Thus, we take $q_i = q_{i'} = 1/2$ and $q_{i,i'} = q_i q_{i'} = 1/4$ 
for all values of $i$ and $i'$.  Suppose further that our knowledge is updated by Boolean messages $M$ and $M'$ according to the formulas $\omega = M$ and $\omega' = M\ \cap \ M'$ with probability one. 
Finally, suppose that $M$ and $M'$ are independent random variables, with $M$ assuming the values 1 and 0 with equal probability, and $M'$ assuming the value 1 with probability $\frac{1}{4}$ and 0 with 
probability $\frac{3}{4}$.  It follows that $p_{i, i'|m, m'} = 1$ for $(i, i', m, m') = (0, 0, 0, 0)$, $(i, i', m, m') = (0, 0, 0, 1)$, $(i, i', m, m') = (1, 0, 1, 0)$, and $(i, i', m, m') = (1, 1, 1, 1)$ and zero otherwise.  We then have
\begin{align*}
\Phi_{\Delta, \Delta'} (\mathcal M, \mathcal M'; \Omega, \Omega')
		&= \sum_{i, i', m, m'} \varphi_{m, m'}  p_{i, i'|m, m'} \log \left(\frac{p_{i,i'|m, m'}}{q_{i,i'}}\right)\\
           &=\left(\frac{1}{2}\right) \left(\frac{3}{4}\right) (1) \log\left[ \frac{1}{\left(\frac{1}{2}\right) \left(\frac{1}{2}\right)}\right] 
                \ \ \textrm{ (the $(i, i', m, m') = (0, 0, 0, 0)$ term)} \\
           &\hphantom{=\ } +\left(\frac{1}{2}\right) \left(\frac{1}{4}\right) (1) \log\left[ \frac{1}{\left(\frac{1}{2}\right) \left(\frac{1}{2}\right)}\right] 
                \ \ \textrm{ (the $(i, i', m, m') = (0, 0, 0, 1)$ term)} \\
           &\hphantom{=\ } +\left(\frac{1}{2}\right) \left(\frac{3}{4}\right) (1) \log\left[ \frac{1}{\left(\frac{1}{2}\right) \left(\frac{1}{2}\right)}\right] 
                \ \ \textrm{ (the $(i, i', m, m') = (1, 0, 1, 0)$ term)} \\
           &\hphantom{=\ } +\left(\frac{1}{2}\right) \left(\frac{1}{4}\right) (1) \log\left[ \frac{1}{\left(\frac{1}{2}\right) \left(\frac{1}{2}\right)}\right] 
                \ \ \textrm{ (the $(i, i', m, m') = (1, 1, 1, 1)$ term)} \\
           &= 2.
\end{align*}

In this example, $p_{i, i'|m, m'} = p_{i' |i, m, m'}$ for all values of $i, i', m$ and $m'$ because $m$ and $m'$ completely determine $i'$ without any need 
to know $i$.  Also, $q_{i'|i} = q_{i'}$
because $q_i$ and $q_{i'}$ are independent.  $\Phi_{\Delta'| \Delta} (\mathcal M'; \Omega' | \mathcal M, \Omega)$ can therefore be calculated 
exactly as above, except that the denominator of the quantity in the logarithm evaluates to 1/2 instead of the 1/4 that appears in the above.  
We conclude that $\Phi_{\Delta'| \Delta} (\mathcal M'; \Omega' | \mathcal M, \Omega) = 1$ in this example.\\

A similar, albeit simpler calculation shows that $\Phi_\Delta (\mathcal M; \Omega) = 1$ in the example, as must be the case, given the above chain rule.
More surprising is the result that 
\begin{align*}
\Phi_{\Delta'} (\mathcal M'; \Omega') &= 3/4\\
                                                      &< \Phi_{\Delta'| \Delta} (\mathcal M'; \Omega' | \mathcal M, \Omega) \\
                                                      &=1.
\end{align*}
In other words, conditioning need not reduce pragmatic information, as it does with Shannon entropy, because the receipt of $M$ can provide a context in which $M'$ has become more informative.

\subsection{Pragmatic Independence of Messages}

It is well known that the joint Shannon entropy of two random variables is the sum of their individual entropies if and only if the random variables are probabalistically independent.  Here, we consider the corresponding situation for
pragmatic information, {\it i.e.} conditions under which 
$$
\Phi_{\Delta, \Delta'} (\mathcal M, \mathcal M'; \Omega, \Omega')=\Phi_{\Delta} ({\mathcal M}; \Omega)  +
\Phi_{\Delta'} ({\mathcal M'}; \Omega')   
$$
holds.  If so, we say that messages $M$ and $M'$ are {\bf pragmatically independent} with respect to their respective decision makers, $\Delta$ and $\Delta'$, with the understanding
that $\Delta$ and $\Delta'$ could coincide.  The following corollary to Theorem 3.2.3 gives a sufficient, but not necessary condition for such pragmatic independence:

\begin{corollary}[Additivity of Pragmatic Information]
If, in addition, to the hypotheses of  Theorem 3.2.4, we have
\begin{align}
q'_{i'|i} = q'_{i'} {\it \quad and \quad} p_{i, i', m, m'} = p_{i, m} p'_{i', m'}  {\it \ for \ all \ } i, i', m, m', 
\label{eq:PragIndependenceCond}
\end{align}
then 
\begin{align}
 \Phi_{\Delta, \Delta'} (\mathcal M, \mathcal M'; \Omega, \Omega')=\Phi_{\Delta} ({\mathcal M}; \Omega)  +
\Phi_{\Delta'} ({\mathcal M'}; \Omega').   
\label{eq:PragIndependence}
\end{align}
\end{corollary}

\begin{proof}
Summing over $i$ and $i'$, we conclude that $\varphi_{m, m'} = \varphi_m \varphi_{m'}$.  Thus, 
\begin{align*}
p_{i, i' | m, m'} &= \frac{p_{i, i', m, m'}}{\varphi_{m, m'}} \\
                       &= \left[\frac{p_{i, m}}{\varphi_{m}}\right] \left[\frac{p'_{i', m'}}{\varphi_{m'}}\right] \\
                       &= p_{i|m} p'_{i'|m'}.  \\
\end{align*}
The result follows, since
\begin{align*}
\Phi_{\Delta, \Delta'} (\mathcal M, \mathcal M'; \Omega, \Omega') 
                       &= \sum_{i, i', m, m'} p_{i, i', m, m'} \log \left(\frac{p_{i,i'|m, m'}}{q_{i,i'}}\right)  \\
                       &= \sum_{i, i', m, m'}p_{i, m} p'_{i', m'} \log \left(\frac{p_{i|m} p_{i'|m'}}{q_{i}q_{i'}}\right)  \\
                       &= \Phi_{\Delta} ({\mathcal M}; \Omega)  + \Phi_{\Delta'} ({\mathcal M'}; \Omega').\\
\end{align*}
\end{proof}
\eqref{eq:PragIndependenceCond} is not necessary, because any message $m$ is pragmatically independent of any other message $m' \in \mathcal M$ for which 
$D_{\Delta'}(\mathbf p_{m'} || \mathbf q) = 0$ for all $m' \in \mathcal M'$.  However, mere probabalistic 
independence of $m$ and $m'$ is not sufficient to guarantee pragmatic independence, as can be seen from the example in 
the previous subsection.  There, we see that
$$
\Phi_{\Delta, \Delta'} (\mathcal M, \mathcal M'; \Omega, \Omega') \neq \Phi_{\Delta} ({\mathcal M}; \Omega)  +
\Phi_{\Delta'} ({\mathcal M'}; \Omega')   
$$
even though $M$ and $M'$ of that subsection are probabilistically independent.

\section{On Disinformation and the Vexed Question of the Value of Information}
The careful reader will note that we have made no attempt thus far to assign a {\it value} to the information in our incoming messages, however valuable they may be to the receiver.  By avoiding any such mention 
of the value of information, we implicitly make the point that the value of a
message is distinct from the amount of pragmatic information it contains.  This point is clear if we remind ourselves 
that pragmatic information is measured in bits, whereas value would be measured in dollars or other 
quantity that would prove useful to a decision maker\footnote{But see below for a discussion of how \cite{Wolpert} 
avoids this problem, albeit in a context that is less general than ours.}.
\\

While the most obvious reason for our reluctance to assign such values is that the value of information is highly context dependent, there are 
several other difficulties with such an assignment.  First, as indicated above, a given message/message ensemble can, 
for a given decision maker, $\Delta$, have positive pragmatic information, yet {\it misinform} $\Delta$.  This follows from our definition of
pragmatic information as information that changes $\Delta$'s views of the state of the world, regardless of whether the new view is any more accurate, better, etc. than some base state.  Nor 
are we guaranteed to have enough context to assign a value to a decision and thus to the messages leading to that decision.  For example, what should be the value of a message
that might prompt a terrible outcome with low probability, but is more likely to prompt more beneficial outcomes?  And how should that value compare to the value of another message that leads uniformly to 
outcomes that are bad, but not that bad?  Even more problematic is our inability to guarantee an assignment of values
$V(m_1), V(m_2), V(m_3)$ to messages $m_1, m_2, m_3 \in \mathcal M$, such that $V(m_1) < V(m_2)$ and $V(m_2) < V(m_3)$ implies $V(m_1) < V(m_3)$.\\
 
Nevertheless, the weaker assumption that messages are either 
\begin{itemize}
\item {\bf pragmatic disinformation (or simply ``disinformation'') with respect to $\Delta$}, 
\item {\bf pragmatically irrelevant with respect to $\Delta$}, or 
\item {\bf pragmatically useful with respect to $\Delta$}
\end{itemize}
is justified by \cite{Wolpert}.  That paper defines its value function 
as the reduction in the amount of information used in describing a system and still maintaining its viability, 
which allows \cite{Wolpert} to consider how a given piece of information may or may not promote 
that viability.   While the discussion in \cite{Wolpert} uses biological examples, their framework is equally relevant for the 
receipt of environmental information that promotes/degrades a receiver's economic 
viability.  More broadly, we might make the above classification based on whether a given message is less useful, the same 
as, or more useful to $\Delta$ than not receiving the message at all.  The above assumption about 
disinformation, irrelevance, or utility of a message is thus a generalization of the value function in 
\cite{Wolpert}.  All we assert is that there is such a function, regardless of whether it pertains to system viability or not.  Since the goal is to transmit useful 
messages, it is helpful to make the following   
\begin{definition}
{\bf Pragmatic noise} is the set of pragmatically irrelevant or disinformative messages with respect to a given $\Delta$.
\end{definition}

We note that \eqref{eq: BasicDef} is the same for pragmatic disinformation and pragmatically useful information, 
because that formula captures only the degree to which the probabilties of $\Delta$'s decisions, 
not the value of these decisions to $\Delta$, have changed.\\

For a given $\Delta$, we can then partition $\mathcal M$ as $\mathcal M = \mathcal I_\Delta \cup \mathcal D_\Delta \cup \mathcal U_\Delta$, where: 
\begin{itemize}
\item $\mathcal D_\Delta$, the set of pragmatically disinformative messages with respect to $\Delta$, is defined by $\mathcal D_\Delta = \{m \in \mathcal M: V(m) < 0\},$
\item $\mathcal I_\Delta$, the set of pragmatically irrelevant messages with respect to $\Delta$, is defined by $\mathcal I_\Delta = \{m \in \mathcal M: V(m) = 0\},$
\item $\mathcal U_\Delta$the,  set of pragmatically useful messages with respect to $\Delta$, is defined by $\mathcal U_\Delta = \{m \in \mathcal M: V(m) > 0\}$.
\end{itemize}
Evidently, $\mathcal N_\Delta = \mathcal I_\Delta \cup \mathcal U_\Delta$ is the set of messages that are
pragmatic noise.
The corresponding amounts of pragmatically irrelevant, disinformative, and useful messages can then be defined as 
$\Phi_{\Delta} (\mathcal I_\Delta; \Omega), \Phi_{\Delta} (\mathcal D_\Delta; \Omega)$, and $\Phi_{\Delta} (\mathcal U_\Delta; \Omega)$, respectively, with
$$
\Phi_{\Delta} (\mathcal I_\Delta; \Omega) =  
    \sum_{m \in \mathcal I_\Delta} \sum_i p_{i, m} \log \left(\frac{p_{i|m}}{q_i}\right), 
$$
$$
\Phi_{\Delta} (\mathcal D_\Delta; \Omega) =  
    \sum_{m \in \mathcal D_\Delta} \sum_i p_{i, m} \log \left(\frac{p_{i|m}}{q_i}\right), 
$$
and
$$
\Phi_{\Delta} (\mathcal U_\Delta; \Omega) =  
    \sum_{m \in \mathcal U_\Delta} \sum_i p_{i, m} \log \left(\frac{p_{i|m}}{q_i}\right).
$$
Clearly, 
$$
\Phi_{\Delta} (\mathcal I_\Delta; \Omega), \Phi_{\Delta} (\mathcal D_\Delta; \Omega), \Phi_{\Delta} (\mathcal U_\Delta; \Omega) \geq 0,
$$
and
$$
\Phi_{\Delta} ({\mathcal M}; \Omega) = \Phi_{\Delta} (\mathcal I; \Omega) +
                                                          \Phi_{\Delta} (\mathcal D; \Omega) +
                                                           \Phi_{\Delta} (\mathcal U; \Omega).
$$
It follows that 
$$
\Phi_{\Delta} (\mathcal I; \Omega), \Phi_{\Delta} (\mathcal D; \Omega), \Phi_{\Delta} (\mathcal U; \Omega) \leq \Phi_{\Delta} (\mathcal M; \Omega).
$$
In particular, if the pragmatic information in some message ensemble $\mathcal M$ is zero, then it cannot contain any pragmatically useful information.\\

\section{Pragmatic Information as ``Free Information,'' By Way of Mutual Information}
Pragmatic information can be interpreted as a kind of information analog of the Hemholtz free energy of statistical mechanics,\footnote{See textbooks such as \cite{StatMech} for an introduction to 
this quantity.} a measure of the amount of energy in a system available to do useful work.  The Hemholtz free energy, $F$, is defined by the formula
$$
F = U - TS,
$$
where $U$ is the total energy of the system, $T$ is the temperature, and $S$ is the system's entropy.  Similarly, 
as we now show, pragmatic information is a measure of the information that is useful in making a decision.
\\

Denote the mutual information between $\mathcal M$ and $\Omega$ as $\mathcal I (\mathcal M; \Omega)$, $p_i = \sum_{m \in \mathcal M} p_{i,m}$ for all $i$,
and $\mathbf p$ as the vector of these marginal probabilities.\footnote{Note the distinction between $\mathbf p_m$ and $\mathbf p$:  the former is
the probability distribution of $\omega$, given a specific $m$, but the latter is the average of these distributions over all $m$'s.}
We then have the following
\begin{theorem} \label{thm:PragInfoDecomp}
\begin{align*}
\Phi _\Delta(\mathcal M; \Omega) &= \mathcal I (\mathcal M; \Omega) + D_\Delta(\mathbf p || \mathbf q) \\
                                                 &= \mathcal H (\mathcal M) - H (\mathcal M| \Omega) + D_\Delta(\mathbf p || \mathbf q). \\
\end{align*}
\end{theorem}
\begin{proof}
\begin{align*}
 \Phi_\Delta (\mathcal M; \Omega) &= \sum_{i, m} p_{i,m} \log \left[\frac{p_{i|m}}{q_i}\right] \cr
                       &= \sum_{i, m}  p_{i,m} \log \left[\frac{p_{i|m}}{p_i}\frac{p_i}{q_i}\right] \cr
                                                  &= \sum_{i, m}  p_{i,m} \log \left[\frac{p_{i|m}}{p_i}\right] + 
								    \sum_i  p_i \log \left(\frac{p_i}{q_i}\right)    \cr
                                                  &   = \mathcal I (\mathcal M; \Omega) + D_\Delta(\mathbf p || \mathbf q) \\
                                                 &= \mathcal H (\mathcal M) - H (\mathcal M| \Omega) + D_\Delta(\mathbf p || \mathbf q), 
\end{align*}
since $\sum_i  p_i \log \left(\frac{p_i}{q_i}\right)$ is the Kullback-Leibler Divergence of $\mathbf p$ and $\mathbf q$.
\end{proof}

We also have the immediate
\begin{corollary}
$
\Phi (\mathcal M; \Omega) \geq \mathcal I (\mathcal M; \Omega),
$
with equality if and only if $\mathbf p = \mathbf q$.\footnote{This corollary resolves a confusion in \cite{Weinberger02}
by showing that pragmatic information and mutual information are equal only in this special case.}
\end{corollary} 

When $\mathbf p = \mathbf q$, the information needed to specify outcome $\omega_i \in \Omega$, {\it i.e.} the amount of information used in making the decision, is 
the total amount of information transmitted by receipt of $m \in \mathcal M$, reduced by $H (\mathcal M| \Omega)$, the amount of information that is unneeded to specify the given outcome,
and thus useless for making the decision.  It is in this sense that pragmatic information is similar in spirit to free energy, the difference between the total energy and the internal energy, {\it i.e.} the information that is 
not available to do useful work.  We can therefore characterize pragmatic information as ``free information''; in fact, it is a generalization of the free information defined in \cite{Crutchfield}.
If, in addition, there is one-to-one relationship between messages and outcomes, $H (\mathcal M| \Omega) = 0$, and $\Phi (\mathcal M; \Omega) = \mathcal H$ 
because all of the information transmitted is used in making the decision. \\

In the more general case in which $\mathbf p \neq \mathbf q$, the term $D_\Delta(\mathbf p || \mathbf q)$ is needed to
update the biases that $\Delta$ had prior to the receipt of $m \in \mathcal M$.  In this case, $\Delta$ is provided with two kinds of information; in addition to $\mathcal H(\mathcal M)$, $\Delta$
is also provided with the information of the update.  The characterization of pragmatic information as free information, {\it i.e.} the difference between total information and ``irrelevant information'',
still makes sense. \cite{Wolpert} comes to similar conclusions by decomposing the syntactic information 
(a.k.a. Shannon entropy) needed for a minimally viable version of a system (analogous to our free information)
and the remainder.
\\

\section{A Reformulation of the Efficient Market Hypothesis}
We now apply the above definitions to make a useful reformulation of the {\it efficient market hypothesis}, a cornerstone of financial economics (see, for example, \cite{Fama}).  According to this hypothesis, current
asset prices, such as stock and bond prices, reflect {\it all} available information, so there is nothing to be gained by poring over charts of price histories or over annual reports.\\
 
However, the assumption that market efficiency is a property of asset price fluctuations suggests that all investors
should find the markets equally inscrutable, a claim belied by the conspicuous long term success of such investors as
Jim Simons and Warren Buffett.  If we take the pragmatic utility of market information to be its ability to
improve its recipient's ability to achieve above-market returns, we can therefore propose the following reformulation:
\begin{hypothesis}
None of the information available to {\rm a given market participant} is pragmatically useful.
\end{hypothesis}

Thus, under this formulation, a market can be efficient to some participants, but not others because information may be pragmatically useful to one market participant 
but not others.\footnote{This is certainly the case, for example, regarding knowledge of short term mispricings that can be exploited by automated trading systems, but not human traders.} 
That the market can be efficient with respect to certain message ensembles but not others is already a staple of the relevant literature.  \cite{Fama}, for example, distinguishes between the ``weak'', 
``semi-strong'', and ``strong'' forms of the hypothesis that, respectively, price histories, all available 
{\it public} information, and all available information, including insider information, are irrelevant in predicting future asset price fluctuations.  Nevertheless, our conceptual framework seems to provide one immediate 
insight: namely, that a market is necessarily efficient in the above sense with respect to a given participant and a given message ensemble if that participant is unable to extract any pragmatically useful information from that ensemble, independent of the performance of others.  One such 
situation is when the amount of incoming market information overwhelms the computational capacities of the participant, as 
the present author has found from personal experience!

\section{Directions for Future Research:  $\Delta$'s Computational Capacity and Beyond}

Implicit in the above are the limits that a decision maker's computational abilities play in accruing pragmatic information.
A sense of the importance of the role of these limits is made clear by an example due to \cite{wow}:  If $\Delta$ has no computational abilities at all, in the sense that any message
received will be ignored, $\mathbf p_m = \mathbf q$ for all messages $m$, and $\Phi_{\Delta} ({\mathcal M}; \Omega) =0$, regardless of the length or complexity of $m$!  In other words, no computation implies no pragmatic information!\\

Furthermore, the details of the computational apparatus that $\Delta$ is assumed to possess are relevant, not just for our efficient markets application, but also for 
biological evolution.  Recall, per \cite{Weinberger02} and references therein, that the quasi-species is a simplfied model of abstract ``replicators'', subject to Darwinian ``survival of the fittest'' in the presence of mutation.   It is the interaction of the quasi-species with its environment that ``decides'' which replicator is, in fact, the fittest.  Therefore, this interaction is the $\Delta$ of the present theory, because it is this interaction that ``decides'' the unknown identity of the fastest growing replicator.  The identity of this replicator is the random variable, $\omega$, and the environmental
details that determine growth and mutation rates are the ``messages'' that the quasi-species environment system processes.  The whole point of the original formulation
of the quasi-species theory \cite{QS} was to demonstrate the existence of an ``error catastrophe", which is a limit to the information that the quasi-species ``computation'' could reliably accrue in any given situation.\\

Also, consider the example of the one armed bandit.  The above calculation assumes that we have the arbitrarily large amount of memory in which we can  store the entire history of play.  The calculation would look quite different if we stored, for example, only the most recent 5 trials.\\

All of the above suggests the need for a systematic exploration of how $\Delta$'s ability to accrue pragmatic information depends on its computational abilities.  One approach might be to formalize these abilities via the
Chomsky Hierarchy familiar to students of abstract automata \cite{Sipser} or the more detailed hierarchy in \cite{Crutchfield}.  The latter paper also suggests that the extensive semigroup theory of automata and his statistical mechanical approach might be 
of use in understanding how $\Delta$ partitions its input messages into its possible updated ``state of the world'' estimates.\\

In the present paper, we have assumed that a single message updates $\Delta$'s estimates of a single, discrete valued random variable $\omega$, the simplest possible setting for our ideas.  For some applications,
however, the assumption that $\omega$ is a continuous random variable, or even a stochastic process, is more appropriate.  It might also be useful to assume that $\Delta$ is responding to a stream of messages by making multiple decisions, perhaps with the goal of rendering each successsive message more pragmatically useful.  We might even imagine that
$\Delta$'s computational machinery is, itself, being updated by these messages, whether via software updates, neural plasticity, 
or some other mechanism.  Such generalizations might provide useful insights into the ubiquitous phenomenon
of collective intelligence,  of which efficient markets are but one example.\footnote{In fact, the recent Santa Fe Institute Conference on Collective Intelligence inspired this work.  See the conference website at
https://www.santafe.edu/info/collective-intelligence-2023/about/.}

\section*{Acknowledgements}
The author would like to acknowledge Profs. Harald Atmanspacher and Herbert Scheingraber for organizing the NATO Advanced Study Institute on Information Dynamics, where the author first learned of the 
need for a theory of pragmatic information; New York University for a partial stipend to attend the Santa Fe Institute
Conference on Collective Intelligence that encouraged him to complete this work; Prof. Atmanspacher, Kevin Atteson, and Keith Lewis for useful discussions; and Wayne Mareci for proofreading the manuscript.
The version of Claude Sonnet 5 available on August 24, 2026 suggested the resemblance to definitions in the literature on decision theory, as well as the need to clarify what is new in the present work.

\end{document}